\documentclass[preprint,12pt]{elsarticle}

\usepackage{amssymb}
\usepackage{amsmath}
\usepackage{multirow}
\usepackage{amsthm}
\usepackage{graphicx, graphics}
\usepackage{subfigure}

\newtheorem{corollary}{Corollary}
\newtheorem{theorem}{Theorem}

\newtheorem{proposition}{Proposition}

\journal{}

\begin{document}

\begin{frontmatter}


\title{A Novel Non-Parametric Approach to Compare Paired General Statistical Distributions between Two Interventions\tnoteref{label1}}
\tnotetext[label1]{This document is a collaborative effort.}
\author[lk]{Kang Li}
\ead{kangli@pku.edu.cn}
\author[fk]{Kai Fan}
\ead{kaifan@pku.edu.cn}

\address[lk]{Department of Probability and Statistics
School of Mathematical Sciences
Peking University}
\address[fk]{Key Laboratory of Machine Perception, MOE,
School of Electronics Engineering and Computer Science,
Peking University}



\begin{abstract}

Despite of many measures applied for determine the difference between two groups of observations, such as mean value, median value, sample standard deviation and so on, we propose a novel non parametric  transformation method based on Mallows distance to investigate the location and variance differences between the two groups. The convexity theory of this method is constructed and thus it is a viable alternative for data of
any distributions. In addition, we are able to establish the similar method under other distance measures, such as Kolmogorov-Smirnov distance. The application of our method in real data is performed as well.

\end{abstract}


\end{frontmatter}


\section{INTRODUCTION}
\label{sec:intro}

The aim of this paper is to provide a method to determine the location and scale relationship between two groups of one-dimensional observations
for two samples, say $\{X_1, X_2, ..., X_n\}$ and $\{Y_1, Y_2, ..., Y_m\}$, such as the responses of two different products on different
subjects, the scores of people on two examinations and so on. Suppose $\{X_i\}$ are independent identity distributed according to $F(\cdot)$ and
$\{Y_i\}$ are independent identity distributed according to $G(\cdot)$, where $F(\cdot)$ and $G(\cdot)$ are two unknown distribution functions.
When the test for normality is not passed, nonparametric analysis methods should be applied. Usually, the mean difference or midian difference
of the two samples is used to determine the location difference and use mean ratio or midian ratio to obtain the scale. They are not reliable
since only a small information of the two samples are extracted and the results are not meaningful.

Based on the idea of location-scale transformation, Freitag, Munk and Vogt \citep{FMV03} has developed an approach to access the structure relationship
between distributions, in which the whole information of the samples are used. However, the problem we concern are the location difference and
scale between two distributions rather than the model structure. That's to say, we have to determine the values of  location difference or scale
or both for any two distributions. Using the same transformation idea, our approach can be described as follows.

Let $\phi(\cdot)$  a linear function, i.e., $\phi(x)=\sigma x+h$, where $\sigma>0$ and $h\in (-\infty, +\infty)$. $d$ is a given measure of
discrepancy between two distributions. Denote the distribution function of $\phi(X_i)$ by $F_1$. Let $D(\sigma, h)=d(F_1, G)$. We want to find
the value $(\sigma_0, h_0)$ which minimize $D(\sigma_0, h_0)$. That's to say, if we transform $\{X_i\}$ to be $\sigma_0X_i+h_0$ , the two groups
of observations are closest and under the "closest" mean we can not tell there are any location difference or scale between them. Therefore, we
can say $\{Y_i\}$ is at least $h_0$ larger than $\sigma_0$ times of $\{X_i\}$.

There are  situations where $(\sigma_0, h_0)$ is not unique. Let  $S=\mbox{argmin}_{(\sigma, h)}D(\sigma, h)$, $S_\sigma=\{\sigma: \exists h, ~
(\sigma,h)\in S\}$ and $S_{h|\sigma}=\{h: (\sigma,h)\in S\}$. Conservatively, at first we can take $\sigma_0=1$ if $1\in S_\sigma$, else
$\sigma_0=\inf S_\sigma$; then take $h_0$ to be the value in $S_{h|\sigma_0}$ satisfying $|h_0|=\inf \{|h|: h\in S_{h|\sigma_0}\}$. When $S$ is
a continues region, it is easy to see that the selected $(\sigma_0, h_0)$ is unique. Therefore, we should find certain $d$ that $S$ is  a
continues region.

Besides, if we let $\sigma\equiv 1$ in $D(\sigma, h)$, the location difference between $F$ and $G$ could be determined. If we let $h\equiv 0$ in
$D(\sigma, h)$, the scale between $F$ and $G$ could be determined. In practice, we could use the empirical distributions of the two group of
data for $F(\cdot)$ and $G(\cdot)$, respectively. The discrepancy measure $d$ we consider in this paper will be focused on Mallows Distance and expand to
Kolmogorov-Smirnov Distance. Mallows Distance was presented in the formulation of statistics framework in 1972, however, an independent physics research work had involved such a related concept a little earlier in 1940s.

The rest of the paper is unfolded as follows. In Section 2 the main results under Mallows Distance for the location transformation, scale
transformation or both are presented, showing that we can uniquely determine the location  and scale relationship between two distributions and
thus Mallows Distance is suitable discrepancy measure to use. In Section 3 the similar results can be obtained under Kolmogorov-Smirnov Distance
but only for location transformation. Section 4 gives the application of this approach to determine the location and scale relationship on real
data.

\section{MATHEMATICAL FORMULA}\label{sec:mth}

\subsection{Definition of Mallows Distance}

In this subsection, we consider the proposed approach under Mallows Distance. Formally, The Mallows $r$-distance (also known as  Wasserstein
$r$-distance) between distributions $F(\cdot)$ and $G(\cdot)$ regarding to random variables $X$ and $Y$, respectively, is defined as
\begin{equation}
d_r(F,G)=\inf_{X,Y}(E|X-Y|^r)^{\frac{1}{r}},
\end{equation}
where the infimum is taken over the set (denoted by $\mathbb{D}_r$) of all joint distributions of $X$ and $Y$ with marginals $F(\cdot)$ and
$G(\cdot)$. Here we require that $X$ and $Y$ have finite $r$th moment, i.e., $E|X|^r<\infty$ and $E|Y|^r<\infty$.


For $r \geq 1$, The Mallows $r$-distance $d_r(F,G)$ has the two properties.
\begin{itemize}
\item Mallows distance $d_r(F,G)$, i.e. satisfies axioms of a metric on $\mathbb{D}_r$.
\item The convergence of distributions in Mallows distance is equivalent to weak convergence plus $r$th moment convergence.(Lavina and Bickel,
 2001)
\end{itemize}

Let $U$ be an uniform random variable, $U \sim Unif(0,1)$, and $F^{-1}(\cdot)$ is the inverse of a distribution function, $F^{-1}(u) = \inf_x
\{x:F(x) \geq u\}$. According to Johnson and Samworth (2005) we know
\begin{equation}\label{eq-u}
\inf_{X,Y}E|X-Y|^r=E|F^{-1}(U)-G^{-1}(U)|^r.
\end{equation}
Equation (\ref{eq-u}) gives an easier computation formula to calculate the distance, that is
\begin{equation}\label{eq-uu}
d_r(F,G)=\left(\int_0^1 |F^{-1}(u)-G^{-1}(u)|^rdu\right)^{\frac{1}{r}} .
\end{equation}
Particularly, when $r=1$, we have a further relationship for computation
\begin{equation}\label{eq-uuu}
d_1(F,G)=\int_0^1 |F^{-1}(u)-G^{-1}(u)|du=\int_{-\infty}^{+\infty}|F(x)-G(x)|dx ,
\end{equation}
which is especially useful when calculating  Mallows 1-distance using empirical distribution for real data, in order to circumvent the unknown
real distribution.

\subsection{Approach under Mallows Distance}

Let $\phi(X)=\sigma X+h$, where $\sigma>0$ and $h \in (-\infty, \infty)$, and $F_1(\cdot)$ be its distribution function, then it is easy to
obtain that $F_1(x)=F(\frac{x-h}{\sigma})$. The purpose of our approach is to find the optimal shift and scale values $(\sigma_0, h_0)$ to
minimizing the Mallows $r$-distance between $F_1(x)$ and $G(x)$, that is
\begin{equation}\label{eq-1}
\arg\min_{\sigma, h} d_r(F_1(x),G(x)) .
\end{equation}

Then the following result can be obtained.
\begin{theorem}
\label{thr-1}
For distribution functions $F(\cdot)$ and $G(\cdot)$, let $F_1(x)=F(\frac{x-h}{\sigma})$ with $\sigma > 0$. Then the Mallows $r$-distance
($r\ge1$) between $F_1(x)$ and $G(x)$, denoted by $D(\sigma, h)$, a function of two variables $\sigma$ and $h$, is a continuous and convex
function on half plane, i.e., for any $0<t<1$, and $\sigma_1 \neq \sigma_2$, $h_1 \neq h_2$, it holds that
\[D(t\sigma_1+(1-t)\sigma_2,th_1+(1-t)h_2) \leq tD(\sigma_1,h_1) + (1-t)D(\sigma_2,h_2) .\]
\end{theorem}
\begin{proof}
It can be easily obtained that $F_1^{-1}(u) = \sigma F^{-1}(u) + h$. From (\ref{eq-uu}), we know
\begin{equation}
D(\sigma,h)=\left(\int_0^1 \left|\sigma F^{-1}(u) + h -G^{-1}(u)\right|^rdu\right)^\frac{1}{r}.
\end{equation}
Then the continuity of the $D(\sigma,h)$ is trivial. Besides,  using Minkowski unequality we have

\begin{align*}
& D(t\sigma_1+(1-t)\sigma_2,th_1+(1-t)h_2)\\
& = \left(\int_0^1 \left|(t\sigma_1+(1-t)\sigma_2)F^{-1}(u)+th_1+(1-t)h_2-G^{-1}(u)\right|^rdu\right)^{\frac{1}{r}}\\
& = \left(\int_0^1 \left|t\left(
\sigma_1F^{-1}(u)+h_1-G^{-1}(u)\right)+(1-t)\left(\sigma_2F^{-1}(u)+h_2-G^{-1}(u)\right)\right|^rdu\right)^{\frac{1}{r}}\\
& \leq t\left(\int_0^1 \left|\left(\sigma_1F^{-1}(u)+h_1-G^{-1}(u)\right)\right|^rdu \right)^{\frac{1}{r}} + (1-t)\left(\int_0^1 \left|\left
(\sigma_2F^{-1}(u)+h_2-G^{-1}(u)\right)\right|^rdu \right)^{\frac{1}{r}} \\
& = tD(\sigma_1,h_1) + (1-t)D(\sigma_2,h_2) .
\end{align*}
\end{proof}

According to the definition of Theorem \ref{thr-1}, scaled parameter $\sigma$ should be greater than zero, but we can easily give an apparent analysis of transformed distribution function if $\sigma \rightarrow 0^+$.

\begin{proposition}
\begin{equation*}
\lim_{\sigma \rightarrow 0^+} F\left(\frac{x+h}{\sigma}\right) = \left\{
\begin{aligned}
   1 &, x > 0  \\
   F(0) &, x=0 \\
   0 &, x<0
   \end{aligned}
\right.
\end{equation*}
\end{proposition}

Also, the Theorem shows that under Mallows $r$-distance ($r\ge 1$) $D(\sigma,h)$ is a convex function of $(\sigma, h)$, thus (\ref{eq-1}) is a
continues region. We can select $(\sigma_0, h_0)$ according to the plan in section \ref{sec:intro}. If we only consider the shifted case or
scaled case, let $\sigma\equiv 1$ or $h\equiv 0$ in $D(\sigma,h)$, then we can obtain the following results.

\begin{corollary}
For distribution functions $F(\cdot)$ and $G(\cdot)$, let $F_1(x)=F(x-h)$. Then the Mallows $r$-distance ($r\ge1$) between $F_1(x)$ and $G(x)$,
denoted by $D(h)$, a function of $h$, is a continuous and convex function on $(-\infty,\infty)$, i.e. for any $0<t<1$, and $h_1 \neq h_2$, it
holds that
\[D(th_1+(1-t)h_2) \leq tD(h_1) + (1-t)D(h_2) .\]
\end{corollary}

\begin{corollary}
For distribution functions $F(\cdot)$ and $G(\cdot)$, let $F_1(x)=F(\frac{x}{\sigma})$ with $\sigma > 0$. The Mallows Distance ($r\ge1$) between
the scaled distribution $F_1(x)$ and $G(x)$, denoted as $D(\sigma)$, a function of $\sigma$, is a continuous and convex function on
$(0,\infty)$, i.e. for any $0<t<1$, and $\sigma_1 \neq \sigma_2$, it holds that
\[D(t\sigma_1+(1-t)\sigma_2) \leq tD(\sigma_1) + (1-t)D(\sigma_2) .\]
\end{corollary}

In order to illustrate $D(\sigma,h)$ may not be strictly convex, let distributions $F(\cdot)$ and $G(\cdot)$ to be
\begin{equation*}
  F(x)=\left\{
   \begin{aligned}
   0 ,&\qquad x \leq -1,  \\
   linear ,&\qquad -1<x<-0.5, \\
   0.5 ,&\qquad -0.5 \leq x \leq 0.5,\\
   linear, &\qquad 0.5<x<1,\\
   1 ,& \qquad x \geq 1,
   \end{aligned}
   \right. \qquad\mbox{and}\qquad
    G(x)=\left\{
   \begin{aligned}
   0, &\qquad x \leq 1,  \\
   linear, &\qquad 1<x<2, \\
   1, &\qquad x \geq 2.
   \end{aligned}
   \right.
\end{equation*}
Actually, $F(\cdot)$ is the uniform distribution over two half unit intervals $[-1,-0.5]$ and $[0.5,1]$, and $G(\cdot)$ is uniform over $[1,2]$.
From (\ref{eq-uuu}) we know $D(h)$ can be calculated via $\int_{-\infty}^{+\infty} |F(x-h)-G(x)|dx$. Then it is easy to verify that $D(h)$
reaches minimum of 0.5 in the entire interval $[-2,-1]$. The optimal shifted value for $\arg\min_{h}D(h)$ is not unique. Therefore, $D(h)$ is
not strictly convex on $(-\infty,\infty)$, nor is $D(\sigma,h)$.

\subsection{Generalization on K-S Distances}

Since our approach is successful under Mallows distance, there is nothing preventing us from exploring other discrepancy measure. Here, we are
able to realize our approach for shifted case under Kolmogorov-Smirnov distance (K-S distance),  $D(F,G)=\sup_x|F(x)-G(x)|$.

For K-S distance and distribution functions $F(\cdot)$ and $G(\cdot)$, our purpose is to find the optimal shift value $h_0$ to minimize
$D(h)=\sup_x|F(x-h)-G(x)|$. Let $F(\infty)=1$ and $F(-\infty)=0$ for distribution $F$. Define $D^+(h)=\sup_x[F(x-h)-G(x)]$ and
$D^-(h)=\sup_x[G(x)-F(x-h)]$, then we have $D(h)=\max\{D^+(h),D^-(h)\}$. And we denote $S=\{h|D(h)=D^+(h)\}$ and $h^*=\inf S$. Due to these
definitions, the statement $-\infty \in S$ is apparently hold and we have the following result.

\begin{theorem}\label{thr-2}
$\mathtt{(1)}$ If $h^* \in S$, then the function $D(h)$ decreases on $(-\infty,h^*]$ and increases on $(h^*,+\infty)$. $\mathtt{(2)}$ If $h^*
\notin S$, then the function $D(h)$ decreases on $(-\infty,h^*)$ and increases on $[h^*,+\infty)$.
\end{theorem}
\begin{proof}
$\mathtt{(1)}$ $\forall h_0 \in S$, if $h_2 \le h_1 \le h_0$, then we have
\begin{equation}\label{eq-2}
D(h_1) \geq D^+(h_1) \geq D^+(h_0)=D(h_0).
\end{equation}

If $D(h_1)=D^-(h_1)$, then $D(h_0) \geq D^-(h_0) \geq D^-(h_1) =D(h_1) \geq D(h_0)$ holds and $ D(h_1)=D(h_0)$. By (\ref{eq-2}), we can obtain
$D(h_1)=D^+(h_1)$. Thus $D(h_1)=D^+(h_1)$ and $ D(h_2) \geq D^+(h_2) \geq D^+(h_1)=D(h_1)$. Therefore, $ D(h)$ decreases on $(-\infty,h^*]$.

Similarly, $\forall h_0 \notin S$, if $h_0 \le h_1 \le h_2$, then we have
\begin{equation}\label{eq-3}
D(h_1) \geq D^-(h_1) \geq D^-(h_0)=D(h_0).
\end{equation}
If $D(h_1)=D^+(h_1)$, then $D(h_0) \geq D^+(h_0) \geq D^+(h_1)=D(h_1) \geq D(h_0)$ holds and $ D(h_1)=D^+(h_0)$. By (\ref{eq-3}), we can obtain
$D(h_1)=D^-(h_1)$. Thus $ D(h_2) \geq D^-(h_2) \geq D^-(h_1)=D(h_1)$. Therefore, $ D(h)$ increases on $(h^*, \infty)$.

$\mathtt{(2)}$ The proof can follow the similar method as $\mathtt{(1)}$ trivially.
\end{proof}

 Theorem \ref{thr-2} shows that our approach under K-S distance can also provide a reasonable, possibly unique location difference  between two
 distributions.

\section{EXPERIMENTS AND SIMULATION}

\subsection{A Real Data Set}

In hair study, we need to assess effets of hair care products in changing hair diameters
 after a period of use. There are two treatments, say $C$ and $F$. The experiments are  conducted as follows.
 There are 30 subjects and each subject use $C$ and $F$ on the left and right head, respectively. There are two study visit,
 baseline and 8 weeks later.
  At each study visit, hair diameters are measured on several hundred
 of hairs on left and right head from a subject.  Comparison between visits is to compare the distributions of hair diameters at two visit point.
 The diameters from one subject often follows  non-traditional
distributions.  For example, a subject at baseline and   8 weeks later hair diameter frequency plot for treatment $C$ are shown in Figure
\ref{fig1} and the distributions are shown in Figure \ref{fig2}. It is of importance to know holistically how much diameters have changed.

\begin{figure}[htbp]
\centering {\subfigure{\label{fig:p1}
\includegraphics[width=0.45\linewidth]{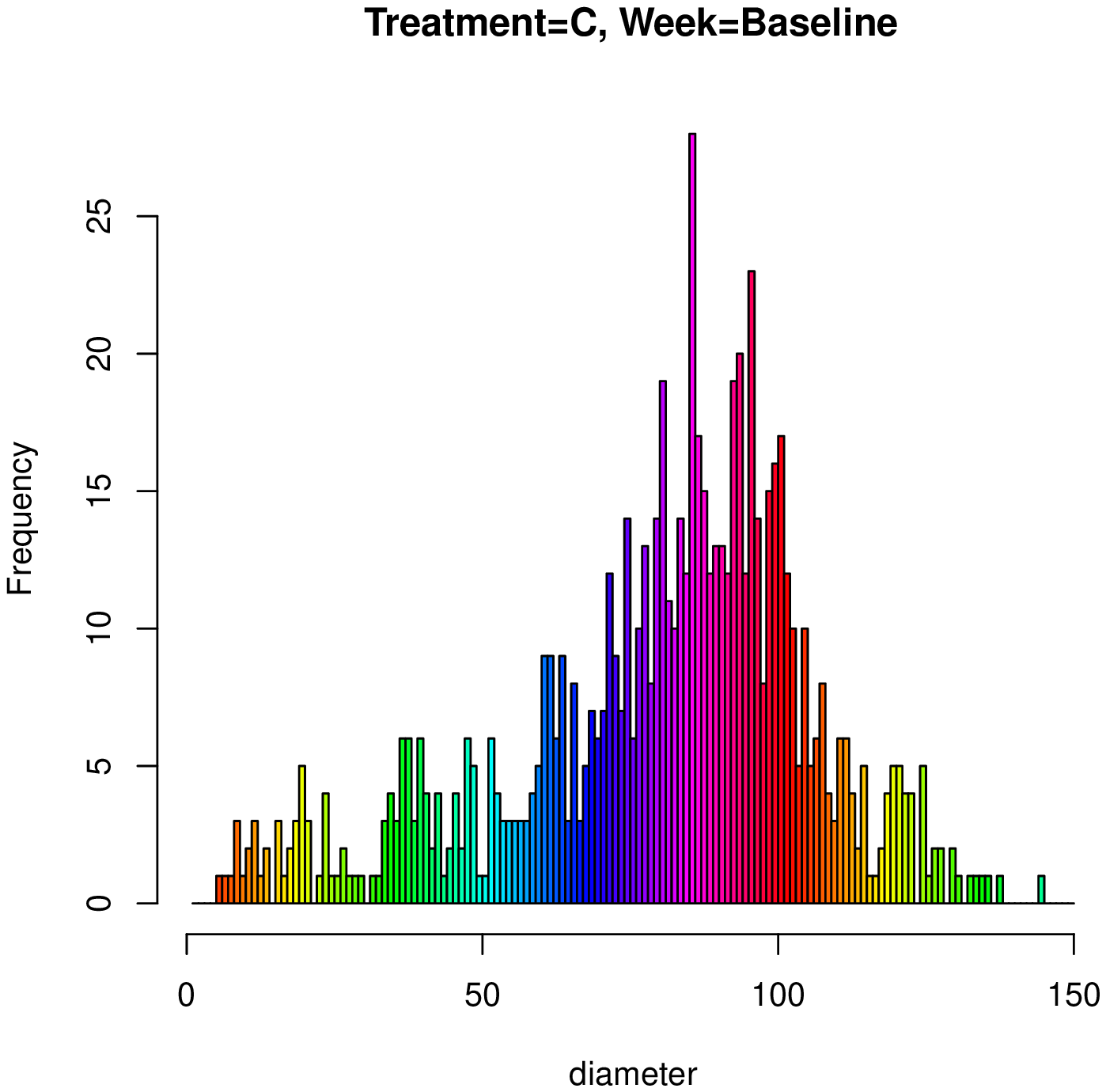}} \qquad 
\subfigure{\label{fig:p2}
\includegraphics[width=0.45\linewidth ]{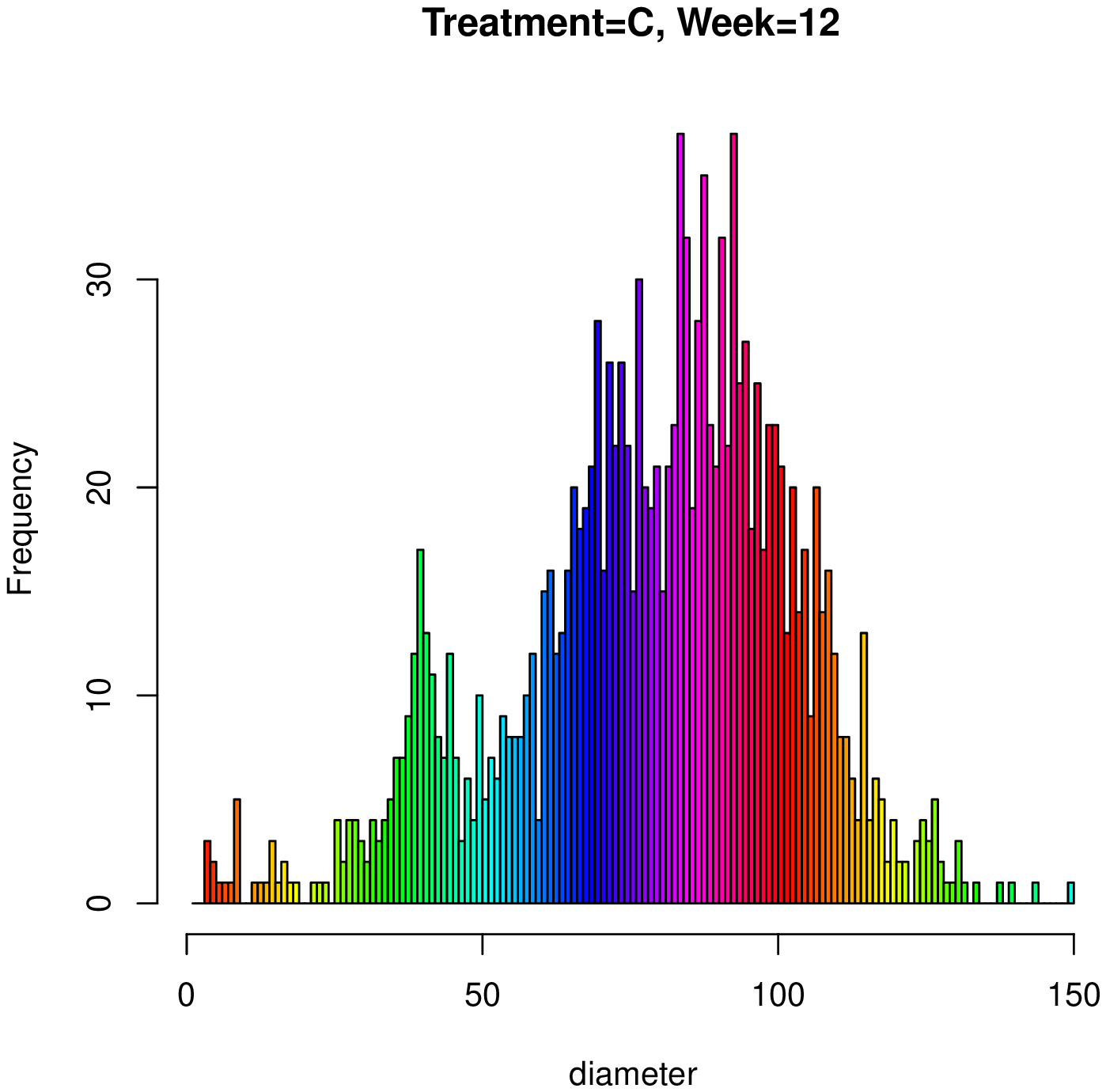}}}
 \caption{\label{fig1}Frequency plot of subject=0001}
\end{figure}

\begin{figure}[htbp]
\centering{
\includegraphics[width=3.6in]{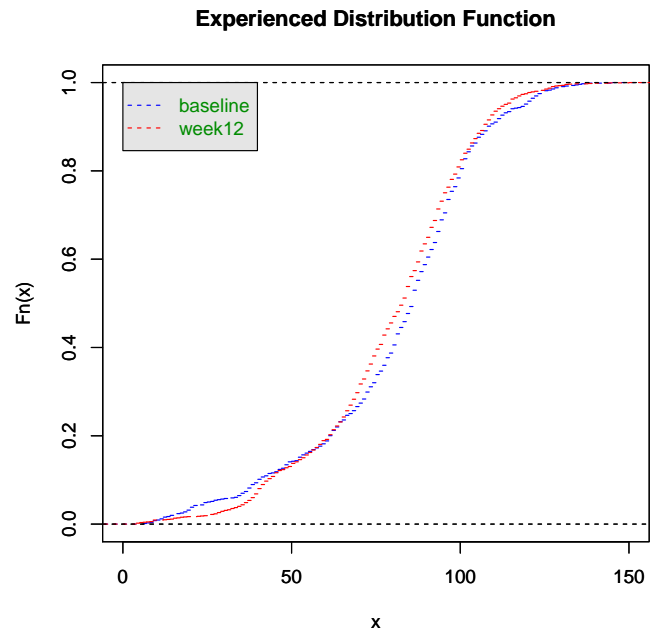}
 \caption{\label{fig2}Distribution of subject=0001}}
\end{figure}

For each subject, the optimal shift amount of the distributions of hair diameters at the two visit point for each treatment under Mallows
distance and K-S distance can be obtained. For instance, the shift plots for $C$ and $F$ of a subject are displayed in Figure \ref{fig3} and
shift plots of all subjects  for $C$ are displayed in Figure \ref{fig4} . The shift corresponding to the minimum distance is the difference
between two distributions, for comparison analysis.

\begin{figure}[htbp]
\centering {\subfigure{\label{fig:p3}
\includegraphics[width=0.45\linewidth]{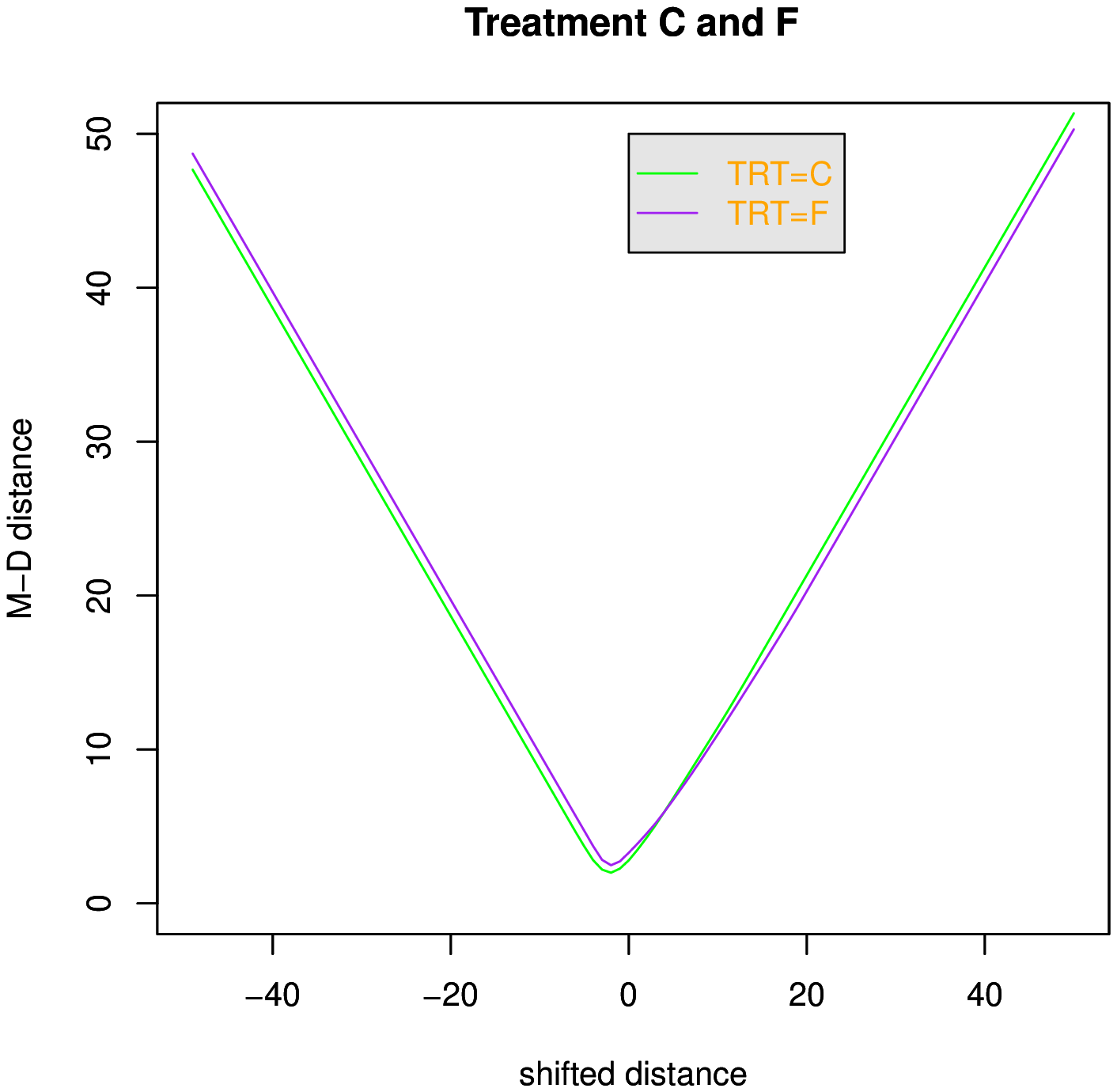}} \qquad \subfigure{\label{fig:p4}\includegraphics[width=0.45
\linewidth ]{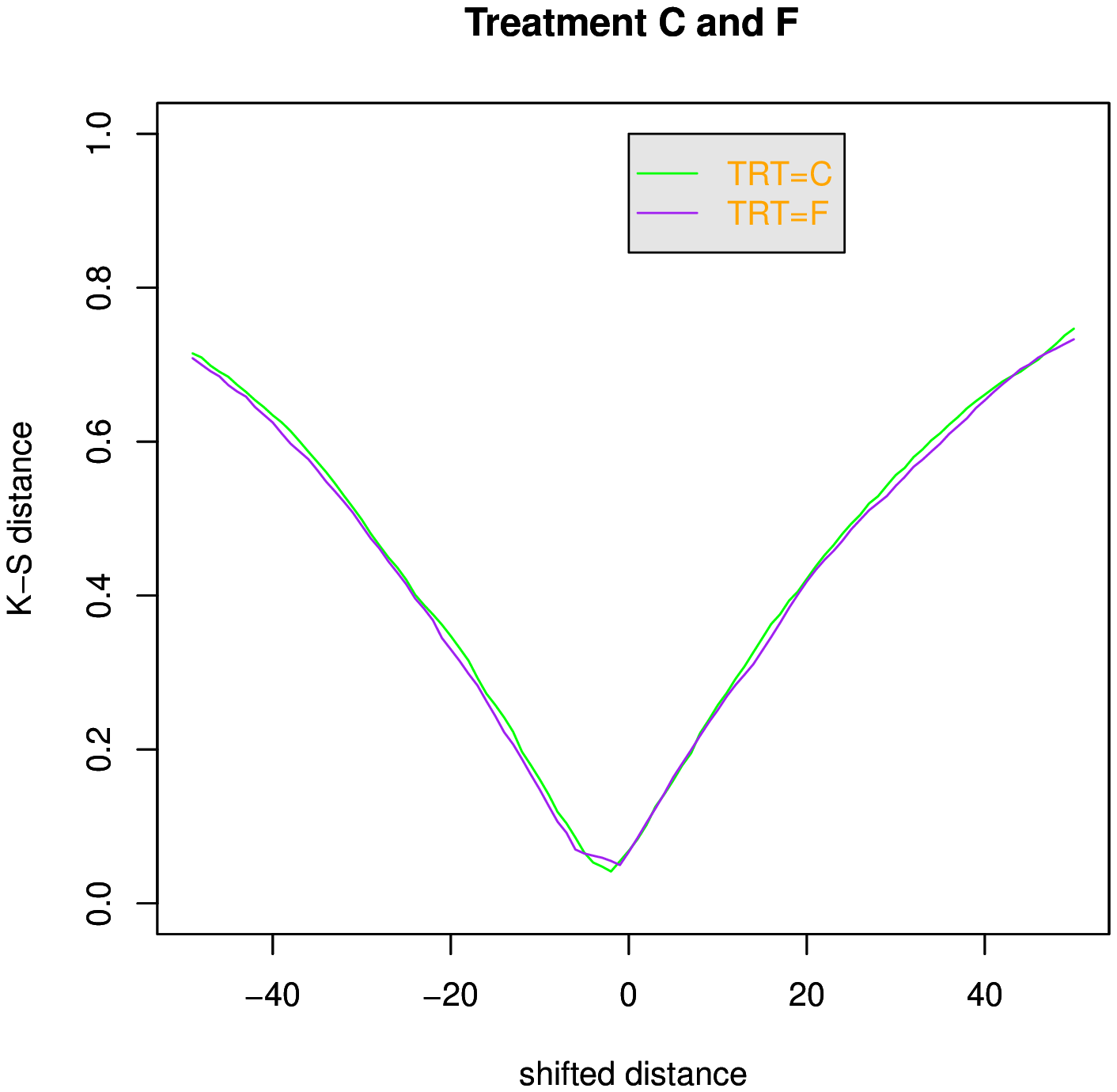}}}
 \caption{\label{fig3}Shift plot of subject=0001}
\end{figure}

\begin{figure}
\centering
\includegraphics[width=3.6in]{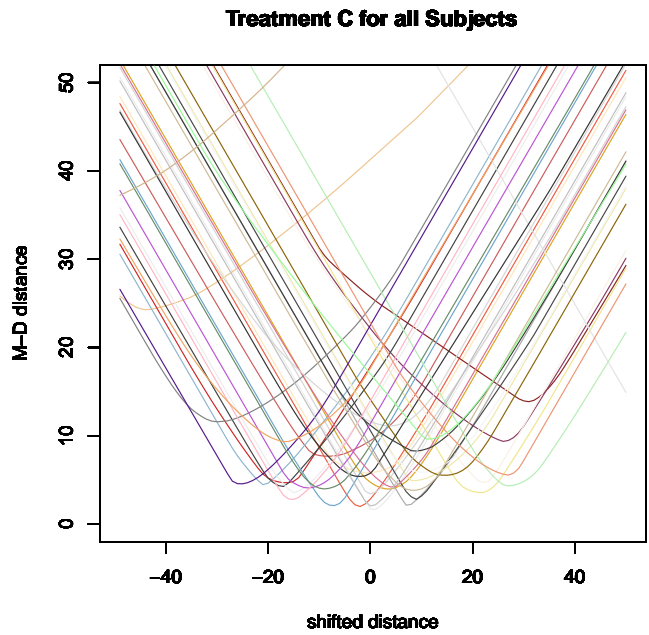}
 \caption{\label{fig4}Shift plot of all subject for $C$}
\end{figure}

Use the optimal values for each subject each treatment as responses and perform Wilcoxon signed rank test on differences in shifts between
treatments for all subjects to detect a difference between the treatments. The results are shown in Table \ref{tab-1}, from which we concludes
that the two treatments are difference at 0.01 level.

\begin{table}[htbp]
  {\caption{\label{tab-1}Comparison of treatments under Distance Shift}}%
  {%
    \begin{tabular}{c|l|l|l|l|l|l}
    \hline
    \bfseries Method & \bfseries n & \bfseries Mean.C & \bfseries Mean.F & \bfseries Mean.(C-F) & \bfseries variance.of.diff & \bfseries p-value \\
    \hline
    M-D &38 & 1.395 & -0.669 & 2.065 & 16.55 &0.0030\\\hline
    K-S & 38 & 1.289 & -0.763 & 2.053 & 19.02 & 0.0025\\
    \hline
    \end{tabular}
  }
\end{table}

\subsection{A Simulation Study}
We apply computer simulation to illustrate our approach on the shift case, scale case, and shift-scale case under Mallows distance or K-S
distance.  Let $ n=100$ and generates two group of independent data $\{X_i\}$ and  $\{Y_i\}$, where $X_i\sim N(\mu_1, \sigma_1^2)$ and $Y_i\sim
N(\mu_2, \sigma_2^2)$, $i=1,2,...,n$. Apply our methods proposed to those data, and calculate the optimal shift value, the optimal scale value
and optimal shift-scale values of the three cases. Repeat this process $M=100$ times and the means and standard errors of those calculated
values are output.

Consider the following four situations: (1) $\mu_1=\mu_2=150, \sigma_1=10, \sigma_2=15$; (2) $\mu_1=150, \mu_2=160, \sigma_1=\sigma_2=10$; (3)
$\mu_1=150, \mu_2=160, \sigma_1=10, \sigma_2=15$. The shift-scale plots for Mallows distance and K-S distance in situations (4) are displayed in
Figure \ref{fig5} and
 the results for all are show in Table \ref{tab-result}. We can see that our approach performed better under Mallows distance than K-S
distance for all situations and cases except for shift case in situation (1), thus our approach is more robust under Mallows distance than K-S
distance. Also, we could demonstrate that the convexity for Mallows distance holds while K-S distance does not, which is consistent with Theorem \ref{thr-1} and \ref{thr-2}. However, both their minimization exist and can be computed.

\begin{figure}[htbp]
\centering {\subfigure{\label{fig:p5}
\includegraphics[width=0.45\linewidth]{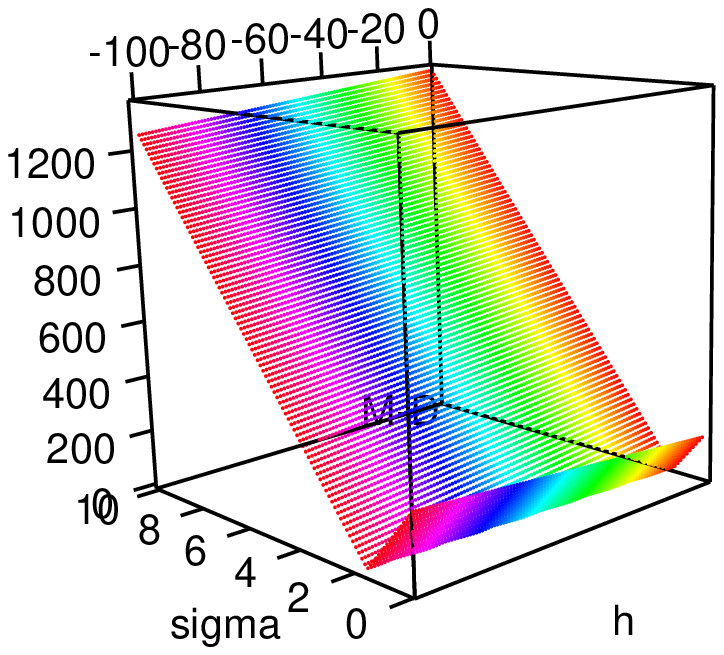}} \qquad \subfigure{\label{fig:p6}
\includegraphics[width=0.45\linewidth ]{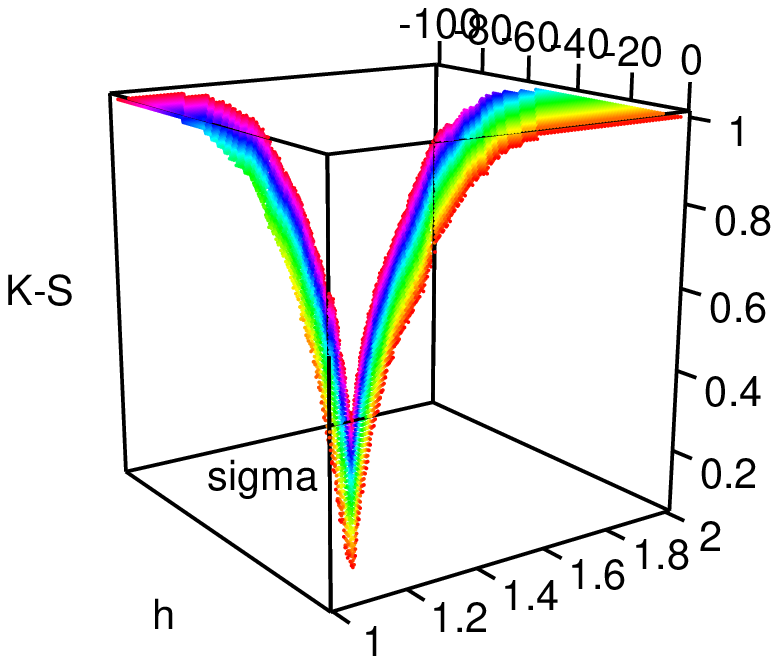}}}
 \caption{\label{fig5}shift-scale plots for Mallows and K-S distance}
\end{figure}

\begin{table}[htbp]
\centering
  {\caption{\label{tab-result}Simulation results: means and standard errors}}
\vspace{2mm}
    \begin{tabular}{c|c|c|c|c|c|c}
    \hline
   \multirow{2}{*}{Situations} & \multicolumn{2}{|c|}{(1)} & \multicolumn{2}{|c|}{(2)} & \multicolumn{2}{|c}{(3)} \\
   \cline{2-7}
         & M-D & K-S & M-D & K-S & M-D & K-S \\\hline
   shift     &  0.11   & 0.04    &  10.01   & 9.93    &  10.01   & 9.85       \\
   $h_0$     &  (2.0)  &(1.9)    &  (1.3)   & (1.4)   &  (2.0)  &  (2.0)     \\\hline
   scale     &  1.00    & 1.00    & 1.09     & 1.09    & 1.09   &1.09       \\
   $\sigma_0$&  (0.0)  & (0.0)   &  (0.02)   & (0.02)  & (0.02) & (0.02)   \\\hline
   shift-scale &       1.46,-69.03 & 1.45, -67.75 & 1.00, 10.00 & 1.01, 8.45 & 1.44,-55.9 & 1.43, -54.93\\
   $\sigma_0, h_0$ & (0.1,15.3) &(0.1,16.5) &(0.1,15.6) &(0.1,19.7) &(0.1,18.1) &(0.1,18.8)  \\\hline
   \end{tabular}
\end{table}

\section{CONCLUSION AND FUTURE WORK}

In this paper, we demonstrated a significant theorem relating to how to measure two distribution within the probabilistic interpretation under
Mallows distance or K-S distance, and a well studied simulation on real data had been implemented for an illustration of this method. The solid
theoretical foundation would be beneficial to others who would have a further understanding or research on Mallows distance measures the
discrepancy between two distributions, especially two similar distributions with inner relationship.

Besides those distances, there might be possibility to use other divergence measures  to be minimized after proper transformations. Comparison
among various underlying divergence measures will be of interest. This is an area of research that we continue to pursue provided available
resources and interests. In addition, A comparison of this approach versus other methods is another topic to be investigated.





\end{document}